\documentclass{article}

\usepackage{dcolumn}
\usepackage{bm}
\usepackage{verbatim}       

\usepackage[pdftex]{graphicx}
\usepackage{amsthm}
\usepackage{amssymb}
\usepackage{bm}
\usepackage{amstext}

\expandafter\let\csname equation*\endcsname\relax

\expandafter\let\csname endequation*\endcsname\relax

\usepackage{amsmath}
\usepackage{amsfonts}
\usepackage{mathrsfs}
\usepackage{cite}

\newcommand{\p}{\partial}

\newcommand{\bd}{\begin{definition}}                
\newcommand{\ed}{\end{definition}}                  
\newcommand{\bc}{\begin{corollary}}                 
\newcommand{\ec}{\end{corollary}}                   
\newcommand{\bl}{\begin{lemma}}                     
\newcommand{\el}{\end{lemma}}                       
\newcommand{\bp}{\begin{proposition}}            
\newcommand{\ep}{\end{proposition}}                
\newcommand{\bere}{\begin{remark}}                  
\newcommand{\ere}{\end{remark}}                     

\newcommand{\bt}{\begin{theorem}}
\newcommand{\et}{\end{theorem}}

\newcommand{\bit}{\begin{itemize}}
\newcommand{\eit}{\end{itemize}}
\newtheorem{theorem}{Theorem}[section]
\newtheorem{corollary}[theorem]{Corollary}
\newtheorem{lemma}[theorem]{Lemma}
\newtheorem{proposition}[theorem]{Proposition}
\theoremstyle{definition}
\newtheorem{definition}[theorem]{Definition}
\theoremstyle{remark}
\newtheorem{remark}[theorem]{Remark}

\begin{document}



\title{A gravitational collapse singularity theorem consistent with black hole evaporation}

\author{E. Minguzzi\thanks{
Dipartimento di Matematica e Informatica ``U. Dini'', Universit\`a
degli Studi di Firenze, Via S. Marta 3,  I-50139 Firenze, Italy.
Corresponding author e-mail: ettore.minguzzi@unifi.it }}



\date{}

\maketitle

\begin{abstract}
\noindent  The global hyperbolicity  assumption present in gravitational collapse singularity theorems is in tension with the quantum mechanical phenomenon of black hole evaporation. In this work I show that the causality conditions in Penrose's theorem can be almost completely removed. As a result, it is possible to infer the formation of spacetime singularities even in the absence of predictability and hence compatibly with quantum field theory and black hole evaporation.
\end{abstract}


The celebrated Penrose singularity theorem \cite{penrose65b,senovilla15} establishes that in a globally hyperbolic spacetime admitting a non-compact Cauchy hypersurface $H$, under the null energy condition, any trapped surface $S$  leads to the formation of a spacetime singularity in its causal future.



If the weak cosmic censorship conjecture is correct, then the spacetime singularity will be hidden behind a horizon so leading to the formation of a black hole. These objects, particularly stationary ones, have then been investigated with the theory of quantum fields over curved spacetimes. A rather robust prediction  states that the black hole will eventually evaporate in a finite, though huge, time \cite{hawking74b,hawking75,wald94}.

It was observed by Kodama \cite{kodama79} and Wald \cite{wald84b} that  the phenomenon of black hole evaporation is in tension with global hyperbolicity and hence with the possibility of predicting the evolution of the spacetime manifold by means of Einstein's equations. Recently Lesourd \cite{lesourd19} has given more precise arguments which show that even causal continuity, a causality condition much weaker than global hyperbolicity, would be violated.

We see that we are in the presence of a loophole argument. In order to infer the formation of singularities we need to assume global hyperbolicity. Unfortunately, one would expect the formation of a black hole and hence its very evaporation. If so the global assumption of global hyperbolicity could not have been satisfied in the first place.
In other words, the hypotheses
of Penrose's theorem seem too strong for a realistic universe, as taking
into account the quantum effects it appears that the existence of a
black hole implies a certain amount of unpredictability.
These observations open the possibility that
trapped surfaces do not necessarily imply the formation of singularities, at least when quantum effects are taken into account.
The situation would be clarified and the problem solved  if we could generalize Penrose's theorem by removing the assumption of  global hyperbolicity.

Much of the discussion about the information loss paradox might have been influenced by the fact that Penrose's theorem requires global hyperbolicity.
It is clear that if the very classical picture of the gravitational collapse had been itself non deterministic  from the outset, then
the prediction of information loss in a quantum field theoretical context would have been perceived as less problematic and the very feeling of a paradox would have been reduced \cite{unruh17}.
We shall indeed show that the very classical picture for gravitational collapse does not require determinism.

Let us discuss whether Penrose's theorem can really be improved.
To start with,  we observe that the Einstein's equation is not used in this theorem, so a deterministic development of the spacetime metric is not required. Only the null energy condition derived from that equation is used. Fortunately, weaker averaged formulations exist that have the same focussing effect on   geodesics but which might be preserved  in a quantum field theoretical context \cite{tipler78,roman88,graham07,wall10,galloway10,fewster11}. Alternatively, when violations of the energy conditions are due to scalar fields one can use weaker versions obtained by replacing the Ricci tensor with the Bakry-Emery Ricci tensor \cite{case10,woolgar13}. All the steps in the proof of Penrose's theorem pass through to the modified Ricci case, as will those of the singularity theorems that we shall present in the following.

It must be mentioned that there already exist singularity theorems that do not assume causality conditions. The most important is Hawking's 1967 theorem \cite{hawking67,hawking73} which establishes that, under suitable energy conditions, expanding partial Cauchy hypersurfaces imply the existence of past singularities. Other variants that get rid of causality assumptions have been considered by Borde \cite{borde85} and Galloway \cite{galloway86b}. It can be observed that all these examples are cosmological in scope. The conditions they impose on spacelike hypersurfaces have global consequences because these hypersurfaces are themselves, so to say,  global. On the contrary, trapped surfaces may form in limited regions of spacetime and getting rid of global causality conditions becomes more difficult.

The  assumption of global hyperbolicity in Penrose's theorem was soon regarded as too strong and quite undesirable already from the classical point of view. It should be recalled that the addition of a small charge or angular momentum to a Schwarzschild  black hole alters its causal structure dramatically, for the maximally extended solutions become the Reissner-Nordstr\"om and Kerr solutions, neither of which is globally hyperbolic.

Hawking and Ellis commented on Penrose's theorem as follows \cite[p.\ 285]{hawking73}

\begin{quote}
 The real weakness of the theorem is the requirement that $H$ be a Cauchy surface. This was used in two places: first, to show that
$(M,g)$ was causally simple which implied that the generators of $\p J^+(S)$ had past endpoints on $S$, [i.e.\ $\p J^+(S)=E^+(S)$] and second, to ensure that under the [global timelike vector field flow-projection] map every point of $\p J^+(S)$ was mapped into a point of $H$.

[Penrose's theorem] does not answer the question of whether singularities occur in physically realistic solutions. To decide this we need a theorem which does not assume the existence of Cauchy surfaces.
\end{quote}

Hawking and Penrose answered these problems with the development of the singularity theorem that bears their name \cite{hawking70}.
Unfortunately, this theorem is in several respects weaker than Penrose's. The singularity might well be to the past of the future trapped surface so, in the context of a spacetime that had origin through a Big Bang singularity, Hawking and Penrose's theorem does not provide any new information for what concerns the formation of a singularity through gravitational collapse. The singularity that it signals could just be the Big Bang singularity. Moreover, the genericity condition there appearing has also been criticized as not always physically justified particularly when there are  inextendible causal curves imprisoned in compact sets or compact Cauchy horizons \cite[Remark 6.22]{minguzzi18b}.


We are left with the strategy of improving Penrose's theorem.
Bardeen \cite{bardeen68} gave an example of null geodesically complete spacetime which satisfies all the assumptions of Penrose's theorem except for global hyperbolicity \cite{hawking73,borde94}. It was obtained through a regularization of the singularity in the Reissner-Nordstr\"om solution.
Considering this spacetime Hawking and Ellis concluded that the global hyperbolicity condition in Penrose's theorem is necessary  \cite{hawking73}.

Nevertheless, Borde \cite{borde94} noticed that Bardeen's example contains a compact partial Cauchy hypersurface, a fact which might indicate that rather than global hyperbolicity, what is essential for the validity of the theorem is a sort of  topological condition  to prevent the compact set $E^+(S)$ from wrapping around the universe (or, stated in more suggestive terms, from swallowing the whole universe). His analysis supported the hope that the global hyperbolicity condition could be weakened in Penrose's theorem.
So far this generalization had been elusive, but with this work we shall indeed show that Borde's intuition was correct.

Since we shall have to replace the global hyperbolicity condition with some weaker conditions, let us have a look at the conformal structure that has been proposed to represent an evaporating black hole \cite{hawking75,hiscock81,wald84b,brown08}, see Fig.\ \ref{eva}.

\begin{figure}[htb]
\centering
\includegraphics[width=5cm]{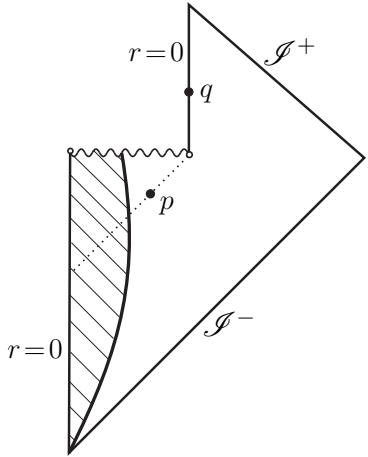}
\caption{The conformal diagram for an evaporating black hole \cite{hawking75}. The points at the edge of the spacelike singularity do not belong to the manifold.} \label{eva}
\end{figure}

As the figure illustrates there are pairs of points such that  $p \in \overline{J^-(q)}$ but $q\notin \overline{J^+(p)}$. This means that the proposed spacetime is not future reflecting though it is past reflecting \cite{minguzzi18b}.

In this work  we first give arguments, complementing those by Lesourd \cite{lesourd19}, which show that in the presence of an evaporating black hole, or more generally of an evaporating null hypersurface (to be defined), future reflectivity cannot hold.
We shall then prove that a version of Penrose's theorem holds just under past reflectivity, cf.\ Theorems \ref{pen} and \ref{pem}.

By the previous result this condition looks as a somewhat technical notion that guarantees that we are not working inside the time dual of an evaporating black hole.
There is the chance that past reflectivity holds  on any reasonable spacetime, for it might represent a sort of entropy condition.
Non-time symmetric entropy conditions are imposed in PDE theory to guarantee the existence of unique solutions \cite{evans98}, so it would be interesting to clarify  if past reflectivity can play a similar role, see also \cite{wald84b,lesourd19}. The validity of past reflectivity would not exclude the possibility that information flows in from a timelike boundary, but such a flow would have to be continuous with no burst or shocks.

Also it can be noticed that we are not going to impose any other causality condition save for the assumption that the trapped surface does not intersect the region of chronology violation (hence the spacetime is non-totally vicious). So we do not impose the compactness of the causal diamonds,  the closure of the causal relation, or the chronology of spacetime. As a consequence, we do not assume the existence of time functions, of arbitrarily small causally convex neighborhoods, or that the chronological relation distinguishes events.

In fact the condition of past reflectivity, though conformally invariant, does not belong to the causal ladder of spacetimes, so its role is not that of preventing nasty influence from infinity nor that of getting rid of almost closed causal curves. It is rather a topological condition on the nature of the causality relation, and hence it places restrictions on {\em how} information propagates on spacetime (it belongs to the so called transverse ladder \cite{minguzzi08b,minguzzi18b}). It is implied by causal continuity, and hence by global hyperbolicity (it is also satisfied in globally hyperbolic spacetimes with timelike boundary \cite[Lemma 3.7]{ake18}),  by the existence of a complete timelike Killing vector field \cite{clarke88}, or by the continuity of the Lorentzian distance \cite[Thm.\ 4.24]{beem96} \cite[Prop.\ 5.2]{minguzzi18b}.

\section{Evaporation and violation of future reflectivity}

We denote with $I$ the chronological relation, and with $J$ the causal relation. The set $E^+(p)=J^+(p)\backslash I^+(p)$ is the {\em future horismos} of $p$. For a subset $S$, $E^+(S):=J^+(S)\backslash I^+(S)$. The inclusion $\subset$ is reflexive. Unless otherwise specified, causal curves $\gamma\colon I \to M$ are future directed and with $\gamma$ we might also denote the image of the curve.

Since in this work we shall make repeated use of the notion of future (and past) reflectivity, it is
 convenient to recall its many equivalent formulations, each giving different insights on its geometrical and physical meaning \cite[Def.\ 4.6]{minguzzi18b}

\begin{definition} \label{hvw}
The spacetime $(M,g)$ is {\em future reflecting} if any of the following equivalent properties holds true. For every $p,q\in M$
\begin{itemize}
\item[(a)]  $p\in \overline{J^{-}(q)} \Rightarrow q\in \overline{J^+(p)} $,
\item[(b)] $ p\in \p {J}^{-}(q) \Rightarrow q\in \p{J}^+(p)$ ,
\item[(c)] $I^-(p)   \subset I^-(q) \Rightarrow I^+(q)\subset I^+(p)$ ,
\item[(d)] $\uparrow\! I^-(p)=I^+(p)$,
\item[(e)]  the volume function $t^+(p)=-\mu(I^+(p))$ is continuous.
\end{itemize}
\end{definition}
Here $\uparrow\! I^-(p):= \textrm{Int} [\cap_{r\in I^{-}(p) } I^+(r)]$ is the {\em common future} while $\mu$ is any finite spacetime measure absolutely continuous with respect to the Lebesgue measures induced by the coordinate charts.

As is well known the black hole region is given by the set $B=M\backslash \overline{I^-(\mathscr{I}^+)}$, where $\mathscr{I}^+$ is the future null infinity in  Penrose's conformal boundary \cite{hawking73}. The achronal set $N:=\p B= \p I^-(\mathscr{I}^+)$ is a locally achronal topologically embedded hypersurface generated by future inextendible lightlike geodesics, namely it is a $C^0$  future null hypersurface \cite{galloway00}  commonly referred as the {\em horizon of the black hole}.
In general by {\em horizon} we shall mean a  $C^0$  future null hypersurface.

Although the reader can keep in mind the relevant  example of a black hole horizon we shall now try to make mathematical sense of what ``evaporating'' or ``evaporated'' could mean for general horizons. Typically what distinguishes an evaporating black hole is the possibility, for an outside observer that has escaped the fate of falling into the black hole, of observing matter fall into the black hole up to a certain proper time. After that instant looking further does not add any new information, for the observer does not see any more matter crossing or approaching the horizon. From the point of view of the observer the gravitational collapse has come to an end. This is in stark contrast with what happens in a Schwarzschild black hole, for which an exterior observer would keep receiving new information from the matter that falls into the black hole.

In order to formalize this concept we need to understand that the observer can look at different regions of the black hole. Let $p\in N$ be a representative point of such a region and let $r\in I^-(p)$. Consider two timelike curves $\gamma$ and $\sigma$,  the former curve $\gamma\colon [0,\infty)\to M$, $\gamma(0)=r$,
$p =\gamma(a)$, $a>0$, represents matter that leaves $r$ and crosses the horizon at $p$, while the latter curve $\sigma\colon [0,\infty)\to M$, $\sigma(0)=r$, $\sigma \cap J^+(N)=\emptyset$, is future inextendible and represents an observer that looks at the infalling matter without being itself causally influenced by the horizon (Fig.\ \ref{fgo}).

\begin{figure}[ht]
\centering
\includegraphics[width=6cm]{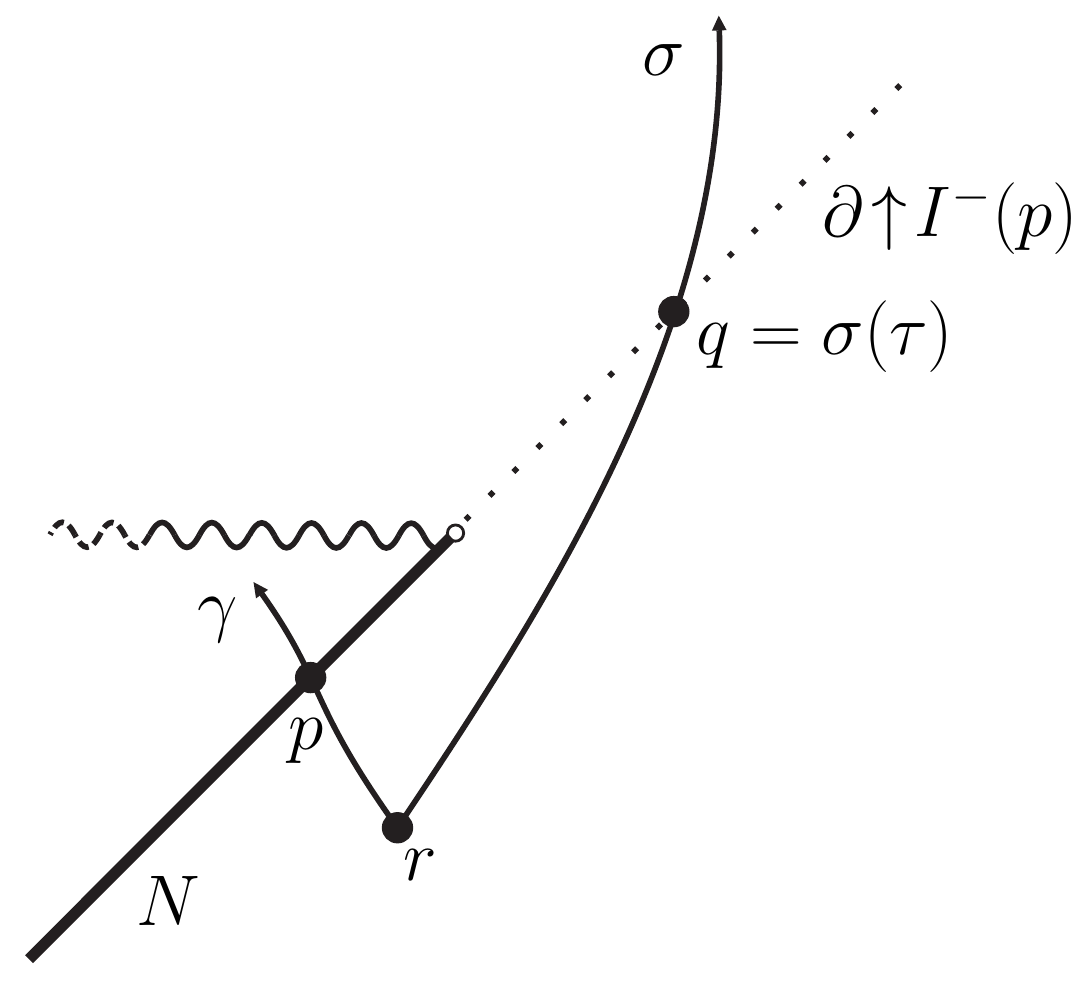}
\caption{Future reflectivity is violated by evaporating future null hypersurfaces.} \label{fgo}
\end{figure}

We are interested in those observers $\sigma$ that can witness the whole falling history, i.e.\ $\gamma([0,a)) \subset J^-(\sigma)$. There are two possibilities, either there is $t>0$ such that $\gamma([0,a)) \subset J^{-}\big(\sigma([0,t])\big)$ or there is not.
\begin{definition}
We say that the  horizon $N$ at the point $p$, or better over the whole generator passing through $p$, has {\em  evaporated} or that it is {\em evaporating} from the point of view of the observer $\sigma$, if there is  $t>0$ such that $\gamma([0,a)) \subset J^{-}\big(\sigma([0,t])\big)$.
\end{definition}
The observer $\sigma$ might wait further time after the event $\sigma(t)$ but will not receive any more information from the matter that has crossed the horizon at $p$.

Let us redefine $\tau=\inf t$ to be the infimum of all the parameters $t$ for which the inclusion holds, and let $q=\sigma(\tau)$. Then for every $p'\in \gamma([0,a))$ we have $ q \in \overline{J^+(p')}$. Let us assume that past reflectivity holds, then $p'\in \overline{J^-(q)}$, and by the arbitrariness of $p'$ and the fact that $\gamma$ is timelike we get $\gamma([0,a))\subset J^{-}(q)$, namely $\tau$ is the least value for which the inclusion holds.

Since $\gamma([0,a))\subset J^{-}(q)$, $\gamma(a)=p$,  we have $p\in \overline{J^{-}(q)}$. However, it cannot be $q\in \overline{J^+(p)}$ since letting $q'=\sigma(\tau')$, $\tau'>\tau$, we would have $q'\gg q$, and hence $q'\in I^+(p)$ in contradiction with the assumption that $\sigma$ does not intersect $J^+(N)$.
We conclude that
\begin{theorem}
Past reflecting spacetimes which contain an evaporating $C^0$ future null hypersurface $N$ are not future reflecting.
\end{theorem}
Notice that every $r'\in \gamma([0,a))$ belongs to $J^{-}(q)$, and since $\gamma$ is timelike, to $I^{-} (q)$, thus $q\in \cap_{r\in I^{-}(p) } I^+(r)$. This is a future set whose interior is the common future $\uparrow  \! \! I^-(p)$, thus $q\in \p [\uparrow\!\! I^-(p)\backslash  I^+(p)]$, where  the set in square brackets is known to be empty under future reflectivity. An observer can also see the horizon evaporate completely, not just the portion near $p$. For that, the observer's  worldline passes through the events belonging to the set $\p [\cap_{r\in I^{-}(N) } I^+(r)]$. Once those events are passed the observer cannot receive any new information from the infalling matter.

\section{Singularities from trapped surfaces}
We recall that the {\em chronology violating set} $\mathcal{C}$ consists of all points $p\in M$ through which there passes a closed timelike curve. No achronal set can intersect $\mathcal{C}$, so on a subset $S$ the condition ``$S$ is achronal'' is stronger than ``$S\cap \mathcal{C}=\emptyset$''.

The reader might be familiar with the notion of {\em future trapped set} which is a non-empty set $S$ such
that $E^{+}(S)$  is non-empty and compact. Under very weak causality conditions  the notion of {\em null araying set} is more convenient \cite{minguzzi18b}.

\begin{definition}
A {\em future lightlike $S$-ray} is a future inextendible  causal curve which starts from $S$ and does not intersect $I^+(S)$.
A non-empty set $S$ is a {\em future null araying set}  if there are no future lightlike $S$-rays.
\end{definition}

Trapped and araying sets are related through the next result \cite[Thm.\ 2.116]{minguzzi18b}.

\begin{theorem} \label{boc}
Let $S$ be a non-empty compact set that does not intersect $\mathcal{C}$. If $S$ is a future null araying set then it is a future trapped set.
\end{theorem}

It is worth recalling that by \cite[Thm.\ 2.100]{minguzzi18b} if $S$ is a non-empty compact set that does not intersect $\mathcal{C}$, then $E^+(S)\cap S=S\backslash I^+(S)\ne \emptyset$, thus $E^+(S) \ne \emptyset$.

\begin{proof}
Suppose $E^{+}(S)$ is not compact then we can find a sequence $q_n\in E^{+}(S)$ for which no subsequence has a limit in $E^+(S)$ (this can happen because $\overline{E^+(S)}$ is non-compact or because $E^+(S)$ is non-closed).
We can find a corresponding sequence $p_n \in S$
such that $q_n \in J^{+}(p_n)$. Passing to a subsequence if
necessary, we can assume $p_n \to p \in S$.
 Consider the causal curves $\sigma_n$ connecting $p_n$ to $q_n$. Necessarily, they do not intersect $I^+(S)$ otherwise $q_n\in I^+(S)$.
By the limit curve
theorem \cite{beem96,minguzzi07c,minguzzi18b} there is a future  limit causal curve $\sigma$
starting from $p$, to which some subsequence of $\sigma_n$, here denoted in the same way, converges uniformly on compact subsets. The curve $\sigma$  does not intersect $I^+(S)$ otherwise for sufficiently large $n$, $\sigma_n$ would intersect $I^+(S)$, which is impossible. In particular, $\sigma$ does not enter $I^+(p)$ and hence it is achronal.  This causal curve  $\sigma$ cannot be a segment, i.e.\ it cannot have  compact domain, for it would connect $p$ to some point  $q\in E^+(S)$, which contradicts the assumption that no subsequence of $q_n$ converges to a point in $E^+(S)$. But then $\sigma$  is future inextendible, hence a lightlike $S$-ray, a contradiction.
\end{proof}

A kind of converse result holds true but it requires additional causality conditions, such as strong causality, that we are not imposing \cite[Thm.\ 2.116]{minguzzi18b}.

The following technical result establishes that past reflectivity is sufficient in order to infer that $E^+(S)$ has no edge. This result is really what makes the generalization of Penrose's theorem possible.

\begin{theorem} \label{edg}
Let $(M,g)$ be past reflecting. If $S$ is a compact
future null araying set that does not intersect $\mathcal{C}$,  then $\p I^+(S)=E^+(S)\ne \emptyset$ and hence $\textrm{edge}({E^{+}(S)})=\emptyset$.
\end{theorem}

 For every open set $U\subset M$, with $I_U$ we  denote the chronological relation in the spacetime $U$ with the induced metric.

\begin{proof}
Suppose not then there is $q\in \p I^+(S)\backslash E^+(S)$. Let $\sigma_n$ be a sequence of timelike curves starting from $S$ and ending at $q_n$ with $q_n\to q$.  Due to the compactness of $S$, the limit curve theorem \cite{minguzzi07c,minguzzi18b} tells us that there are either a continuous causal curve connecting some $r\in S$ to $q$, which is impossible since it would entail $q\in J^+(S)$ and hence $q\in E^+(S)$, a contradiction, or there is a future inextendible continuous causal curve $\sigma^r$ with starting point  $r\in S$ to which some subsequence of $\sigma_n$, here denoted in the same way, converges uniformly on compact subsets in a suitable parametrization. Since $S$ is a future null araying set, $\sigma^r$  cannot be a future lightlike $S$-ray, hence it enters $I^+(S)$. Let $b\in \sigma^r$, such that $U\ni b$ is an open neighborhood contained in $I^+(S)$. Let $p\in I_U^-(b)$, then since $b$ is a limit point of the sequence, for sufficiently large $n$, $\sigma_n$ intersects $I_U^+(p)$, thus $q_n\in I^+(p)$ and $q\in \overline{I^+(p)}$, and by past reflectivity $p\in \overline{I^-(q)}$, thus as $p\in I^+(S)$, $q\in I^+(S)$, a contradiction. 
 \end{proof}

\begin{figure}[ht]
\centering
\includegraphics[width=11cm]{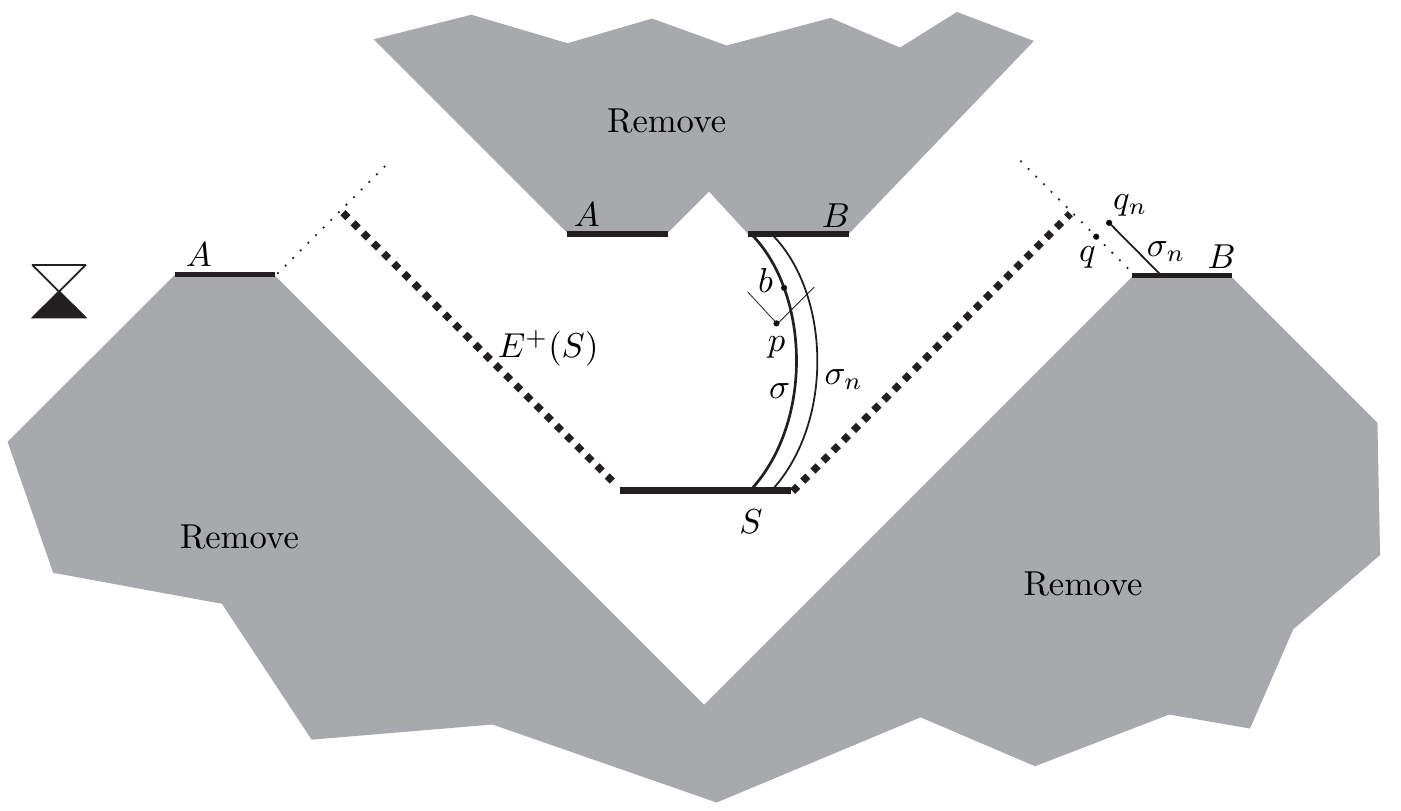}
\caption{The figure illustrates the proof by contradiction of Thm.\ \ref{edg}: suppose there is $q\in \p I^+(S)\backslash E^+(S)$..., then past reflectivity is violated, there is $p\in M$ such that $q\in \overline{J^+(p)}$ but $p\notin \overline{J^-(q)}$. Moreover, it shows that past reflectivity cannot be weakened to the property `transitivity of $\bar J$'.} \label{tra}
\end{figure}

\begin{remark}
Figure \ref{tra} is obtained by removing the shaded regions from Minkowski 1+1 spacetime, and by identifying the remaining sides of the two segments $A$ (and similarly for the two segments $B$). The spacetime is  such that $\bar J$ is transitive (moreover, it is causally easy hence stably causal). Still it is non-past reflecting. While $S$ is a compact achronal null araying set, we have that $\p I^+(S)\ne E^+(S)$ and that $E^+(S)$ has edge. Thus past reflectivity cannot be weakened to the property `transitivity of $\bar J\, $' in Thm.\ \ref{edg} (it is known that past reflectivity implies the transitivity of $\bar J$, cf.\ \cite[Remark 4.14]{minguzzi18b}).
\end{remark}

In order to capture the idea of trapped set that swallows the whole universe, as described in \cite[p.\ 265]{hawking73} and \cite{borde94} we introduce the next definition.

\begin{definition}
A compact non-empty
set $S$ that does not intersect $\mathcal{C}$ is said to have an {\em unavoidable} or {\em swallowing} future horismos if there exists an open neighborhood $U$ of $E^+(S)$ such that $I^-_U(E^+(S))\subset \textrm{Int}D^-(E^+(S))$.
\end{definition}

In fact under this condition an observer, represented by an inextendible causal curve, that were to pass through a neighborhood of the horismos $E^+(S)$ would be forced to intersect it and hence to fall into its causal influence.

\begin{theorem} \label{mvo}
Let $(M,g)$ be
past reflecting.
Every compact
future null araying set $S$ that does not intersect $\mathcal{C}$ has an unavoidable future horismos $E^+(S)$. In fact this horismos is also compact and coincident with $\p I^+(S)$. Finally, $E^+(S)\subset \textrm{Int} D(E^+(S))$.
\end{theorem}

Once again figure \ref{tra} shows that the assumption of past reflectivity cannot be replaced by the weaker notion of transitivity of $\bar J$.


\begin{proof}
By Thm.\ \ref{boc} $E^+(S)$ is non-empty and compact, moreover from Thm.\  \ref{edg} we know that $E^+(S)=\p I^+(S)$ (in what follows we use only these properties and the fact that $S$ is a future null araying set that does not intersect $\mathcal{C}$, past reflectivity is not used). First we show that there is an open set $U\supset E^+(S)$ such that $I^-_U(E^+(S))\subset D^-(E^+(S))$ from which the inclusion of the statement follows by the openness of $I^-_U(E^+(S))$. If the inclusion $I^-_U(E^+(S))\subset D^-(E^+(S))$ does not hold for any open neighborhood $U$ of $E^+(S)$, then by considering a sequence of nested relatively compact neighborhoods $U$ we deduce that we can find a sequence of future inextendible causal curves $\sigma_n$ starting from $p_n \in I^-(E^+(S))$ and not intersecting $E^+(S)$ such that $p_n\to p\in E^+(S)$. By the limit curve theorem there exists a future inextendible continuous causal curve $\sigma$ starting from $p$ to which a subsequence of $\sigma_n$, here denoted in the same way, converges uniformly on compact subsets. This continuous causal curve $\sigma$ enters $I^+(S)$ because $S$ is future null araying (notice that $p\in J^+(S)$).

Notice that no point of $\overline{I^+(S)}$ can belong to $I^{-}(\p {I}^+(S))$, for if $x\ll y$, with $x\in \overline{I^+(S)}$ and $y\in \p {I}^+(S)$, by the openness of $I$, $y\in I^+(S)$, a contradiction.

As $\sigma$ intersects $I^+(S)$,  for sufficiently large $n$ the same is true for $\sigma_n$, which then has to pass from $I^{-}(E^+(S))$ to the disjoint set $\overline{I^+(S)}$  hence intersecting the boundary of the latter set, $\p I^+(S)=E^+(S)$, a contradiction since $\sigma_n$ does not intersect $E^+(S)$.

Let us prove the last statement.
Let $p\in E^+(S)$ and let $\gamma$ be a timelike curve, $\gamma(0)=p$. From the just proved swallowing property we know that for some $\epsilon>0$, $\gamma((-\epsilon,0))\subset \textrm{Int} D^-(E^+(S))$. We want to show that for some $\delta>0$, $\gamma(\delta)\in \textrm{Int} D^+(E^+(S))$, then, as $\textrm{Int} D(E^+(S))$ is causally convex in $M$ \cite[Prop.\ 3.43]{minguzzi18b}, we can conclude that $p\in  \textrm{Int} D(E^+(S))$. Suppose that $\delta=1/n$ does not work for any $n$, then we can find a past inextendible causal curve $\sigma_n$ with future endpoint $p_n:=\gamma(1/n)$ not intersecting $E^+(S)$. By the limit curve theorem there is a past inextendible causal curve $\sigma$ ending at $p$, to which a subsequence  $\sigma_{n_k}$ converges. The curve $\sigma$ cannot intersect $I^+(S)$ because $p \notin I^+(S)$.
However, $\sigma$ cannot escape  $\overline{I^+(S)}$ in the past direction for, by convergence, $\sigma_{n_k}$ would do the same for sufficiently large $k$,  which would mean that it passes in the future direction from  $M\backslash \overline{I^+(S)}$ to $I^+(E^+(S))\subset \overline{I^+(S)}$, thus intersecting  $\p I^+(S)=E^+(S)$, a contradiction since $\sigma_{n_k}$ does not intersect $E^+(S)$.
We conclude that $\sigma$ is a past inextendible causal curve entirely contained in the compact achronal set $\p I^+(S)=E^+(S)$. By  \cite[Prop.\ 3.4]{minguzzi07f} \cite[Sec.\ 2.3]{minguzzi18b} there is a lightlike line $\eta$ entirely contained in $E^+(S)$ to which $\sigma$ accumulates.\footnote{Introducing a complete Riemannian metric $h$, this result follows from the application of the limit curve theorem to a sequence of $h$-arc length parametrized  causal curves of the form $\sigma_k(t)=\sigma(t-a_k)$ where $a_k$ is a  diverging sequence. Since  the $h$-arc length of the curves on both sides of the new origins $\sigma(-a_k)$ diverges, the limit curve that passes   through an accumulation point of $\{\sigma(-a_k)\}$ is inextendible rather than just past inextendible, see the references for details.} Let $q\in \eta$. As $q\in J^+(S)$ we can join $S$ to $q$ with a causal curve entirely contained in $E^+(S)$ and then continue following $\eta$ in the future direction so obtaining a future    inextendible causal curve which is actually a lightlike $S$-ray as it does not intersect $I^+(S)$, a contradiction with the araying property of $S$.
\end{proof}
%
%

\begin{definition}
 We say that a spacetime is {\em (spatially) open} if it does not contain  a compact spacelike  hypersurface.
\end{definition}

The next result will clarify that Penrose's theorem can be deduced from our Thm.\ \ref{pen}.

\begin{proposition}
Every spacetime admitting a non-compact Cauchy hypersurface (hence globally hyperbolic) is open.
\end{proposition}

\begin{proof}
Suppose that there is a compact spacelike hypersurface $\Sigma$. The argument is as in Penrose's theorem and makes use of the flow of a global $C^1$ timelike vector field $v$
to project $\Sigma$ to  the Cauchy hypersurface $H$. As $\Sigma$ is spacelike such a projection is open, but $\Sigma$ is compact and the projection is continuous, thus the projection of $\Sigma$ is compact and hence closed. But $H$ is connected (because $M$ is connected and  it splits topologically as $M\simeq \mathbb{R} \times H$) thus $\Sigma$ and $H$ are homeomorphic, a contradiction as $H$ is non-compact.
\end{proof}

\begin{remark}
There are other classes of open spacetimes. Standard stationary spacetimes \cite{caponio11,caponio14e} are of the form $\mathbb{R}\times S$, where  the $\mathbb{R}$-fibers are the integral curves of a timelike Killing vector field $\p_t$. For $S$ non-compact they can be easily shown to be open. The existence of a timelike Killing vector field $\p_t$ implies that they are reflecting \cite{clarke88}\cite[Thm.\ 4.10]{minguzzi18b}. Spacetimes that are  past reflecting but not future reflecting can be easily constructed from them as follows: let $A$ be a spacelike hypersurface with boundary, and remove from the spacetime the set $\cup_{t\ge 0} \varphi_t(A)$ where $\varphi$ is the flow of the Killing vector field.
Unfortunately,  due to the timelike Killing field, one cannot find trapped surfaces inside these spacetimes \cite{mars03}. One would need to use more general  metrics, e.g.\ \cite[Thm.\ 5.9]{caponio14e}, in which $\p_t$ is not necessarily timelike.  Indeed, in the presence of a Killing vector field Penrose's theorem and our theorem \ref{pen} apply only if there is a region in which the Killing vector is spacelike \cite{mars03}, see also Remark \ref{nex}.
\end{remark}

\begin{theorem} \label{alt}
Let $(M,g)$ be a
past reflecting and  open spacetime. Then it does not admit
compact future null araying sets that do not intersect $\mathcal{C}$.
\end{theorem}

\begin{proof}
By contradiction, let $S$ be such an araying set.
By Thm.~\ref{edg} $E^+(S)$ is an achronal boundary hence a locally Lipschitz topological hypersurface, and by Thm.\ \ref{boc} it is compact.
By Theorem \ref{mvo} $\textrm{Int} D(E^+(S))\ne \emptyset$, thus this open set is a globally hyperbolic spacetime once endowed with the induced metric \cite[Prop.\ 6.6.3]{hawking73}\cite[Thm.\ 3.45]{minguzzi18b}, so we can find a spacelike Cauchy hypersurface $H$ for it \cite{bernal03}. By Thm.\ \ref{mvo} we have also $E^+(S) \subset \textrm{Int} D(E^+(S))$.
Let $v$ be a global $C^1$ timelike vector field. Its flow in $\textrm{Int} D(E^+(S))$ establishes a bijection between $H$ and $E^+(S)$. Indeed, given a point $p\in H$ the integral curve of $v$ passing through $p$ intersects $E^+(S)$ as $H\subset \textrm{Int} D(E^+(S))$. Similarly, given a point $q\in E^+(S)$ the  integral curve of $v$ passing through $q$ intersects $H$ as $H$ is a Cauchy hypersurface for $\textrm{Int} D(E^+(S))$. The established map is really continuous with inverse continuous thus $H$ and $E^+(S)$ are homeomorphic (this is the standard argument used in  Penrose's theorem). We conclude that $H$ is also compact and that $(M,g)$ is not spatially open, a contradiction.
\end{proof}

\begin{remark}
The previous result is optimal in the sense that past reflectivity cannot be weakened to the property ``$\bar J$ is transitive''. An example is provided by Fig.\ \ref{tra} where two vertical future intextendible timelike geodesics starting from $I^+(S)$ are removed. One has to be chosen between the two $A$ segments with starting point quite close to $S$, and similarly the other has to be chosen between the two $B$ segments. Their purpose is to make the spacetime open by removing  the compact spacelike hypersurfaces that intersect $A$ and $B$.
\end{remark}

A {\em trapped surface} is a spacelike compact codimension 2 manifold without boundary such that $\theta^+,\theta^-<0$ all over $S$, where these quantities are the divergences of the two future lightlike geodesic congruences issued from $S$. The {\em null convergence condition} states that $Ric(n,n)\ge 0$, for every null vector $n$, where $Ric$ is the Ricci tensor.


We obtain the following theorem which drops the global hyperbolicity assumption from Penrose's theorem.

\begin{theorem} \label{pen}
Let $(M,g)$ be a  past reflecting spacetime which is  open and satisfies the null convergence condition. Suppose that it admits a future trapped surface $S$ such that $S\cap \mathcal{C}=\emptyset$,
then it is future null geodesically incomplete.
\end{theorem}

\begin{proof}
Let us assume future null geodesic completeness.
First let as prove that $S$ is future null araying.
If not $S$ is  a trapped surface that admits a future lightlike $S$-ray.
This
ray must start perpendicularly to $S$ otherwise it would be entirely
contained (save for  the starting point) in $I^{+}(S)$. Thus
this ray belongs to one of the two congruences of converging
lightlike geodesics issued from $S$. By the Raychaudhuri equation
and null geodesic completeness, this geodesic reaches a focal
point (to the surface $S$) and hence it enters $I^{+}(S)$, a contradiction
that proves that $S$ is future null araying.
However, by Thm.\ \ref{alt} this is impossible,
which proves that the spacetime is future null geodesically incomplete.
\end{proof}

\begin{remark} \label{nex}
The generality of this theorem can  already be appreciated in the study of exact solutions.  Let us show that it predicts the singularity of the  Reissner-Nordstr\"om spacetime. Indeed, it is clear from the Penrose conformal diagram \cite{hawking73} that this solution is reflecting and spatially open. Its stress-energy tensor is of electromagnetic type thus it satisfies the null convergence condition \cite[Ex.\ 5.1]{misner73}. Any spherically symmetric surface between its Cauchy and interior horizons is trapped (moreover, there are no trapped surfaces in the exterior region because there is a timelike Killing vector there), thus our theorem predicts a geodesic null singularity which is indeed found at $r=0$. On the contrary Penrose's theorem does not apply because this spacetime is not globally hyperbolic.

The anti de Sitter–Schwarzschild spacetime is another non-globally hyperbolic spacetime which admits trapped surfaces and for which our theorem applies.

In the Kerr–Newman solution \cite{carter68,hawking73} for $a^2+e^2<m^2$, $a\ne 0$, we have  two horizons given by $r=r_-$ and $r=r_+$, that separate the spacetime into three types of regions. Region III, $r<r_-$, is the chronology violating set, region II, $r_-<r<r_+$, contains trapped surfaces \cite{hawking73,senovilla11}, while region $I$, $r_+<r$ is asymptotically flat. If the closure of the chronology violating set, $r\le r_-$, is removed from the manifold we are left with a globally hyperbolic spacetime \cite{carter68}, thus  our theorem and Penrose's  could be applied. Unfortunately,  the singularity they detect is simply due to the removal of region III, for  the null geodesics leaving the trapped surface enter it. If the region III is kept Penrose's theorem cannot be applied because the spacetime is no more globally hyperbolic, and our theorem cannot be applied because past reflectivity is violated at the past boundary of the chronology violating set (cf.\ proof of \cite[Prop.\ 4.26]{minguzzi18b}). However, our theorem can possibly be applied if only the region $r\le 0$ is removed because in this case past reflectivity might hold (a proof can be tricky because the available Penrose diagrams for this solution are really relative to some section of it).

Finally,  Newman \cite{newman89} gave an example, inspired by  G\"odel solution, of  a complete spatially open spacetime  that satisfies the null convergence condition and admits a trapped surface in $\mathcal{C}$. His example shows that the condition $S\cap \mathcal{C}\ne \emptyset$ cannot be dropped from Theorem \ref{pen}. Moreover, if in his 2+1-dimensional example we remove the region $t\ge 0$ and take a slightly smaller trapped ring (which indeed remains trapped) we obtain an example of application of Theorem \ref{pen} for which the spacetime is past reflective but not chronological and in which there is a trapped surface that does not intersect $\mathcal{C}$. If instead, we remove just the region $r\ge \sinh^{-1}(1/\sqrt{2})$, $t=0$, and take a slightly smaller trapped ring we obtain an example of trapped surface $S$ whose lightlike generators are complete and for which $E^+(S)$ has edge. Here past reflectivity is violated, thus past reflectivity cannot be  dropped in Thm.\ \ref{edg}.
\end{remark}



Theorem \ref{mvo} has the following corollary which shows that Borde intuition was correct.

\begin{theorem} \label{pem}
Let $(M,g)$ be a  past reflecting spacetime which  satisfies the null convergence condition. Suppose that it admits a future trapped surface $S$ such that $S\cap\mathcal{C}=\emptyset$,
then it is either future null geodesically incomplete or the horismos $E^+(S)$ is compact, unavoidable and actually coincident with $\p I^+(S)$.
\end{theorem}

\begin{proof}
Suppose that $(M,g)$ is future null geodesically complete.  By arguing as in the proof of Thm.\ \ref{pen} we obtain that $S$ is a null araying set, thus from  Theorem \ref{mvo} we get that  the horismos $E^+(S)$ is compact, unavoidable and actually coincident with $\p I^+(S)$.
\end{proof}
 We conclude that unless the horismos swallows the whole universe a singularity is formed in the future of the trapped surface.

\section{Conclusions}

In Penrose's singularity theorem we find the following two conditions: (a) the spacetime is globally hyperbolic, and (b) there is a non-compact Cauchy hypersurface. We have shown that they can be weakened as follows: (a') the spacetime is past reflecting, and (b') the spacetime is spatially open.

By dropping the global hyperbolicity condition we have shown that determinism  is not necessary in order to infer that trapped surfaces lead to spacetime singularities. The classical picture one gets of the gravitational collapse is therefore consistent with that of quantum field theory on curved backgrounds, where one would have to renounce global hyperbolicity in any case, due to the phenomenon of black hole evaporation.
Our result helps to resolve some of the tension between quantum field theory and general relativity by showing that both can accommodate coherent descriptions of black hole formation and evaporation in non-globally hyperbolic spacetimes.

We have also proved a singularity theorem in which condition (b) is replaced by a condition which states that the horismos $E^+(S)$ of the trapped surface does not swallow the whole universe (in which case investigation through cosmological singularity theorems would have been more appropriate), so proving a conjecture by Borde.

\section*{Acknowledgements}
Useful conversations with Ivan P. Costa e Silva, Jos\'e Luis Flores, Miguel S\'anchez and Jos\'e Senovilla, and useful comments by the referees  are acknowledged. Work partially supported by GNFM-INDAM.


\end{document}